\documentclass[sigconf]{acmart}
\usepackage[utf8]{inputenc} 
\usepackage[T1]{fontenc}    
\usepackage{url}            
\usepackage{booktabs}       
\usepackage{amsfonts}       
\usepackage{nicefrac}       
\usepackage{microtype}      
\usepackage[page,title]{appendix}
\usepackage{comment}
\usepackage{enumitem}
\usepackage{threeparttable}

\usepackage{bm}
\usepackage{dsfont}
\usepackage{amsmath}
\usepackage{amsthm}
\newtheorem{theorem}{Theorem}

\usepackage{graphicx}
\usepackage{subfigure}
\usepackage{capt-of}

\usepackage{hyperref}
\usepackage{comment}

\usepackage{wrapfig}       
\usepackage{diagbox}
\usepackage{multirow}
\usepackage{makecell}

\usepackage{pifont}
\usepackage[ruled]{algorithm2e}

\def \bfE {\mathbb{E}}

 \def \cU {\mathcal{U}}
 
\def \cL{\mathcal{L}}

 \def \cD{\mathcal{D}}

\newcommand{\indep}{\perp \!\!\! \perp}
\newcommand{\notindep}{\not \! \perp \!\!\! \perp}

\AtBeginDocument{%
  }

\setcopyright{acmcopyright}
\copyrightyear{2023} 
\acmYear{2023} 
\setcopyright{acmlicensed}\acmConference[WWW '23]{Proceedings of the ACM Web Conference 2023}{May 1--5, 2023}{Austin, TX, USA}
\acmBooktitle{Proceedings of the ACM Web Conference 2023 (WWW '23), May 1--5, 2023, Austin, TX, USA}
\acmPrice{15.00}
\acmDOI{10.1145/3543507.3583495}
\acmISBN{978-1-4503-9416-1/23/04}

\begin{document}

\title{Balancing Unobserved Confounding with a Few Unbiased Ratings in Debiased Recommendations}

 \author{Haoxuan Li}
\affiliation{
  \institution{Peking University}
  \country{China}}
 \email{hxli@stu.pku.edu.cn}

 \author{Yanghao Xiao}
\affiliation{
  \institution{{University of Chinese Academy of Sciences}}
  \country{China}}
 \email{xiaoyanghao22@mails.ucas.ac.cn}

 \author{Chunyuan Zheng}
\affiliation{
  \institution{University of California, San Diego}
  \country{USA}}
 \email{czheng@ucsd.edu}

\author{Peng Wu}
 \authornote{Corresponding author.}
\affiliation{
  \institution{Beijing Technology and Business University}
  \country{China}}
 \email{pengwu@btbu.edu.cn}


%
\renewcommand{\shortauthors}{Trovato et al.}

\begin{abstract}
Recommender systems are seen as an effective tool to address information overload, but it is widely known that the presence of various biases makes direct training on large-scale observational data result in sub-optimal prediction performance. In contrast, unbiased ratings obtained from randomized controlled trials or A/B tests are considered to be the golden standard, but are costly and small in scale in reality. To exploit both types of data, recent works  proposed to use unbiased ratings to correct the parameters of the propensity or imputation models trained on the biased dataset. However, the existing methods fail to obtain accurate predictions in the presence of unobserved confounding or model misspecification. In this paper, we propose a theoretically guaranteed model-agnostic balancing approach that can be applied to any existing debiasing method with the aim of combating unobserved confounding and model misspecification. 
The proposed approach makes full use of unbiased data by alternatively correcting model parameters learned with biased data, and adaptively learning balance coefficients of biased samples for further debiasing. Extensive real-world experiments are conducted along with the deployment of our proposal on four representative debiasing methods to demonstrate the effectiveness.

\end{abstract}



\maketitle

\section{Introduction}

Recommender systems (RS) are designed to accurately predict users' preferences and make personalized recommendations. In recent years, many studies have focused on deep learning for rating predictions, aiming to fit the collected data using proper deep model structures~\cite{cheng2016wide, he2017neural, wang2017deep, guo2017deepfm}. Despite the ease of collection and large scale of observed ratings, it is known that such data always contain various biases and fail to reflect the true preferences of users~\cite{Chen-etal2020, Wu-etal2022, li-etal-MR2023}. For instance, users always choose the desired items to rate, which causes the collected ratings to be missing not at random, and training directly on those would leads to long-tail effects~\cite{abdollahpouri2017controlling} and bias amplification~\cite{wang2021deconfounded}.

To perform debiasing directly from the biased ratings, previous studies can be summarized into three categories:
\begin{itemize}
    \item Inferring missing and biased ratings, then replacing them using pseudo-labels~\cite{marlin2012collaborative, Steck2010}. However, the data sparsity in RS and unobserved features of users and items make it difficult to estimate those missing values accurately.
    \item Estimating the probability of a rating being observed, called propensity, then reweighting the observed data using the inverse propensity~\cite{Schnabel-Swaminathan2016, saito2020ips, Wang-Zhang-Sun-Qi2019,saito2020doubly}. 
    However, the unobserved confounding, affecting both the missing mechanism and the ratings, makes it fail to completely eliminate the biases.
    \item Modeling missing mechanisms and data generating process using generative models~\cite{liang2016modeling}. However, it may leads to violation of model specification and data generating assumptions in the presence of unobserved variables, resulting in biased estimates.
\end{itemize}
It can be summarized that these methods would lead to biased estimates in the presence of unobserved confounding or model misspecification. To mitigate the effects of unobserved confounding, Robust Deconfounder (RD) proposes an adversarial learning that uses only biased ratings~\cite{ding2022addressing}. Specifically, RD assumes that the true propensity fluctuates around the nominal propensity and uses sensitivity analysis to quantify the potential impact of unobserved confounding. However, the assumption cannot be empirically verified from a data-driven way, and it is essential to relax the assumptions while reducing the bias due to the unobserved confounding and model misspecification.

In contrast to observational ratings, uniform ratings are considered the golden standard and can be obtained from A/B tests or randomized controlled trials (RCTs), but harm users' experience and are costly and time-consuming~\cite{Gilotte-etal2018, Gruson-etal2018}. Due to its small scale property, it is impractical to train prediction models directly on unbiased ratings. Recent studies propose to use a few unbiased ratings for the parameter selection of the propensity and imputation models using bi-level optimization, which has a more favorable debiasing performance compared with the RCT-free debiasing methods~\cite{Wang-etal2021, Chen-etal2021}. However, 
we show that using unbiased ratings only to correct propensity and imputation model parameters still leads to biased predictions, in the presence of unobserved confounding or model misspecification. This motivates a more sufficient use of the unbiased ratings to combat the effects of unobserved confounding.

In this paper, we propose a model-agnostic approach to balance unobserved confounding with a few unbiased ratings. Different from the previous debiasing methods, our approach enlarges the model hypothesis space to include the unbiased ideal loss. The training objective of the balancing weights is formalized as a convex optimization problem, with balancing the loss estimation between biased and unbiased ratings as constraints. Through theoretical analysis, we prove the existence of the global optimal solution. Then, we propose an efficient training algorithm to achieve the training objectives, where the balancing weights are reparameterized and updated alternatively with the prediction model. Remarkably, the proposed balancing algorithm can be applied to any exsiting debiased recommendation methods. The main contributions of this paper are summarized as follows.
\begin{itemize}
    \item  We propose a principled balancing training objective with a few unbiased ratings for combating unmeaseured confounding in debiased recommendations.
    \item To optimize the objectives, we propose an efficient model-agnostic learning algorithm that alternatively updates the balancing weights and rating predictions.
    \item Extensive experiments are conducted on two real-world datasets to demonstrate the effectiveness of our proposal.
\end{itemize}

\section{Preliminaries}
Let $\{u_1, u_2, \dots, u_M\}$ be a set of $M$ users, $\{i_1, i_2, \dots, i_N\}$ be the set of $N$ items, and $\cD = \{(u_m, i_n)\mid m=1, \dots, M; n = 1, \dots, N\}$ be the set of all user-item pairs. Denote $\mathbf{R}=\{r_{u, i}\mid (u, i)\in \cD\}\in \mathbb{R}^{|\mathcal{D}|}$ be a true rating matrix, where $r_{u, i}$ is the rating of item $i$ by user $u$. However, users always selectively rate items based on their interests, resulting in observed ratings, denoted as $\mathbf{R}^\mathcal{B}\in\mathbb{R}^{|\mathcal{B}|} (\mathcal{B}\subseteq \cD)$, are missing not at random and thus biased. For a given user-item pair $(u, i)$, let $x_{u, i}$ be the feature vector of user $u$ and item $i$, such as user gender, age, and item attributes, etc. Let $o_{u, i}$ be the binary variable indicating whether $r_{u, i}$ is observed $o_{u, i}=1$ or missing $o_{u, i}=0$. Given the biased ratings $\mathbf{R}^\mathcal{B}$, the prediction model $\hat r_{u, i}=f(x_{u, i}; \theta)$ in the debiased recommendation aims to predict all true ratings accurately. Ideally, it can be trained by minimizing the prediction error between the predicted rating matrix $\mathbf{\hat R}=\{\hat r_{u, i}\mid (u, i)\in \cD\}\in \mathbb{R}^{|\mathcal{D}|}$ and the true rating matrix $\mathbf{R}$, and is given by
\begin{equation}\label{ideal}
				\cL_{ideal}(\theta)  
				=  \frac{1}{|\cD|}  \sum_{(u,i) \in \cD} \delta(r_{u,i}, \hat r_{u, i})=\frac{1}{|\cD|}  \sum_{(u,i) \in \cD} e_{u, i},\end{equation}where $\delta(\cdot, \cdot)$ is a pre-specified loss, and $e_{u,i}$ is the prediction error, such as the  squared loss  
$e_{u,i} = (\hat r_{u,i}  - r_{u,i} )^2$.

For unbiased estimates of the ideal loss in Eq. (\ref{ideal}), previous studies proposed to model the missing mechanism of the biased ratings $\mathbf{R}^\mathcal{B}$. Formally, the probability $p_{u, i} = \operatorname{Pr}(o_{u,i}=1 | x_{u,i})$ of a user $u$ rating an item $i$ is called propensity. The inverse probability scoring (IPS) estimator~\cite{Schnabel-Swaminathan2016} is given as
    \begin{align*}
    \cL_{IPS}(\theta) = \frac{1}{|\cD|} \sum_{(u,i) \in \cD}  \frac{ o_{u,i}e_{u,i} }{ \hat p_{u, i} }, 
     \end{align*}
where $\hat p_{u, i}=\pi(x_{u, i}; \phi_p)$ is an estimate of the propensity $p_{u, i}$, and the IPS estimator is unbiased when $\hat p_{u, i}=p_{u, i}$. The doubly robust (DR) estimator~\cite{Wang-Zhang-Sun-Qi2019, saito2020doubly} is given as
\begin{align*}
         \cL_{DR}(\theta) = \frac{1}{|\cD|} \sum_{(u,i) \in \cD} \Big [ \hat e_{u,i}  +  \frac{ o_{u,i} (e_{u,i} -  \hat e_{u,i}) }{ \hat p_{u, i} } \Big ], 
     \end{align*} 
where $\hat e_{u, i}=m(x_{u, i}; \phi_e)$ fits the prediction error $e_{u, i}$ using $x_{u, i}$, i.e., it estimates $g_{u, i}=$ $\mathbb{E}\left[e_{u, i} \mid x_{u, i}\right]$, and DR has double robustness, i.e., it is unbiased when either $\hat e_{u, i} = g_{u, i}$ or $\hat p_{u, i} = p_{u, i}$.

In industrial scenarios, randomized controlled trials or A/B tests are considered to be the golden standard, and users might be asked to rate randomly selected items to collect unbiased ratings, denoted as $\mathbf{R}^\mathcal{U}\in\mathbb{R}^{|\mathcal{U}|} (\mathcal{U}\subseteq \cD)$. The ideal loss can be estimated unbiasedly by simply taking the average of the prediction errors over the unbiased ratings
\begin{equation*}
		\cL_{\mathcal{U}}(\theta)  
				=  \frac{1}{|\mathcal{U}|}  \sum_{(u,i) \in \mathcal{U}}  e_{u, i}\approx\cL_{ideal}(\theta).\end{equation*}However, unbiased ratings are costly and small in scale in reality. To exploit both types of data, recent works  proposed to use unbiased ratings to correct the parameters of the propensity or imputation models trained on the biased dataset.
Learning to debias (LTD)~\cite{Wang-etal2021} and AutoDebias~\cite{Chen-etal2021} propose to use bi-level optimization, using unbiased ratings $\mathbf{R}^\mathcal{U}$ to correct the propensity and imputation model parameters, and then the prediction model is trained by minimizing the IPS or DR loss estimated on the biased ratings $\mathbf{R}^\mathcal{B}$. Formally, this goal can be formulated as
\begin{align}
\phi^* &=\arg \min _\phi \mathcal{L}_{\mathcal{U}}\left(\theta^*(\phi) ; \mathcal{U}\right) \label{eq-5}
\\
\text { s.t. } \theta^*(\phi) &=\arg \min _\theta \mathcal{L}_{\mathcal{B}}(\theta, \phi ; \mathcal{B}) \label{eq6},
\end{align}
where $\mathcal{L}_{\mathcal{B}}$ is a pre-defined loss on the biased ratings, such as IPS with $\phi=\{\phi_p\}$, DR with $\phi=\{\phi_p, \phi_e\}$, and AutoDebias with an extra propensity that $\phi=\{\phi_{p1}, \phi_{p2}, \phi_e\}$. The bi-level optimization first performs an assumed update of $\theta(\phi)$ by Eq. (\ref{eq6}), then updates the propensity and imputation model parameters $\phi$ by Eq. (\ref{eq-5}), and finally updates the prediction model parameters $\theta$ by Eq. (\ref{eq6}).
\section{Proposed Approach}
We study debiased recommendations given biased ratings with a few unbiased ratings. Different from previous studies~\cite{Schnabel-Swaminathan2016, Wang-Zhang-Sun-Qi2019, Chen-etal2021, MRDR, Wang-etal2021, Dai-etal2022, li-etal-MR2023, TDR, SDR}, we consider there may be unmesaured confounding in the biased ratings, making the unconfoundedness assumption no longer hold. In Section \ref{sec:3.1}, we show that simply using unbiased ratings to perform model selection of propensity and imputation does not eliminate the bias from unobserved confounding and model misspecification. In Section \ref{sec:3.2}, we propose a balancing training objective to combat the unobserved confounding and model misspecification by further exploiting unbiased ratings. In Section \ref{sec:3.3}, we propose an efficient model-agnostic algorithm to achieve the training objective.
\subsection{Motivation}\label{sec:3.1}
First, the unbiasedness of IPS and DR requires not only that learned propensities or imputed errors are accurate, but also the unconfoundedness assumption holds, i.e., $o_{u, i} \indep e_{u, i}\mid x_{u, i}$. However, there may exist unobserved confounding $h$, making $o_{u, i} \notindep e_{u, i}\mid x_{u, i}$ and $o_{u, i} \indep e_{u, i}\mid (x_{u, i}, h_{u, i})$. Let $\tilde p_{u, i}=\operatorname{Pr}(o_{u, i}=1\mid x_{u, i}, h_{u, i})$ be the true propensity, then the nominal propensity $p_{u, i} \neq \tilde p_{u, i}$, and Lemma \ref{lemma:1} states that the existing IPS and DR on $\mathbf{R}^\mathcal{B}$ are biased estimates of the ideal loss in the presence of unobserved confounding.
\begin{lemma}\label{lemma:1}
The IPS and DR estimators are biased in the presence of unobserved confounding, even the learned propensities and imputed errors are accurate, i.e., $\hat p_{u, i}=p_{u, i}$, $\hat e_{u, i}=g_{u, i}$, then 
\[
\bfE[\cL_{IPS}(\theta)] - \bfE[\cL_{ideal}(\theta)] = \operatorname{Cov}\left(\frac{o_{u, i}-p_{u, i}}{p_{u, i}}, e_{u, i}\right)\neq0,
\]
and
\[
\bfE[\cL_{DR}(\theta)] - \bfE[\cL_{ideal}(\theta)] = \operatorname{Cov}\left(\frac{o_{u, i}-p_{u, i}}{p_{u, i}}, e_{u, i}-g_{u, i}\right)\neq0.
\]
\end{lemma}
\begin{proof}
For DR estimator, if $\hat{p}_{u, i}=p_{u, i}, \hat{e}_{u, i}=g_{u, i}$, we have
$$
\begin{aligned}
\bfE[\cL_{DR}(\theta)]={}& \mathbb{E}\left[e_{u, i}+\frac{o_{u, i}-p_{u, i}}{p_{u, i}}\left(e_{u, i}-g_{u, i}\right)\right] \\
={}& \bfE[\cL_{ideal}(\theta)]+\mathbb{E}\left[\frac{o_{u, i}-p_{u, i}}{p_{u, i}}\left(e_{u, i}-g_{u, i}\right)\right] \\
={}& \bfE[\cL_{ideal}(\theta)]+\operatorname{Cov}\left(\frac{o_{u, i}-p_{u, i}}{p_{u, i}}, e_{u, i}-g_{u, i}\right).
\end{aligned}
$$
The last equation follows by noting that \[\mathbb{E}\left[\frac{o_{u, i}-p_{u, i}}{p_{u, i}}\right]= \mathbb{E}\left[\mathbb{E}\left\{\frac{o_{u, i}-p_{u, i}}{p_{u, i}}\mid x_{u, i}\right\}\right]=0,\]
and $\mathbb{E}[e_{u, i}-g_{u, i}]=0$. In the presence of hidden confounding, $\operatorname{Cov}((o_{u, i}-p_{u, i}) / p_{u, i}, e_{u, i}-g_{u, i}) \neq 0$. 
The conclusions of the IPS estimator can be obtained directly from taking $g_{u, i}=0$ in DR. 
\end{proof}
In addition, the existing methods using bi-level optimization, as shown in Eq. (\ref{eq-5}) and Eq. (\ref{eq6}), simply uses unbiased ratings for parameter tuning of the propensity and imputation models. It follows that the prediction models in hypothesis space $\mathcal{H}_{\phi}=\{\cL_{\mathcal{B}}(\theta, \phi)\mid \phi \in \Phi\}$ are as a subset of DR, where $\Phi$ is the parameter space of $\phi$. Though the unbiased ratings correct partial bias, in the presence of unobserved confounding or model misspecification, i.e., $\cL_{ideal}\notin\mathcal{H}_{\phi}$, it is still biased due to the limited $\mathcal{H}_{\phi}$. 
\begin{proposition}\label{th2}
The IPS and DR estimators are biased, in the presence of (a) unobserved confounding  or (b) model misspecification.
\end{proposition}
Proposition \ref{th2} concludes the biased property of IPS and DR in the presence of unobserved confounding or model misspecification.
\subsection{Training Objective}\label{sec:3.2}
To combat unobserved confounding and model misspecification on biased ratings, we propose a balancing approach to fully leverage the unbiased ratings for debiased recommendations. First, when there is no unobserved confounding, we have \[\bfE[\mathcal{L}_{\mathcal{B}}(\theta, \phi ; \mathcal{B})]=\bfE[\mathcal{L}_{\mathcal{U}}\left(\theta(\phi) ; \mathcal{U}\right)].\]To obtain unbiased estimates in the presence of unmeasured confounding or model misspecification, 
we propose to enlarge the hypothesis space to include the ideal loss, from $\mathcal{H}_{\phi}$ to $\mathcal{H}_{Bal}=\{\boldsymbol{w}^{T}\cL_{\mathcal{B}}(\boldsymbol{x}; \theta, \phi)\mid \phi \in \Phi, \boldsymbol{w}\in \mathbb{R}^{|\mathcal{D}|}\}$, where $\cL_{\mathcal{B}}(\boldsymbol{x}; \theta, \phi)\in \mathbb{R}^{|\mathcal{D}|}$ consists of the contribution of $(u, i)$ to $\mathcal{L}_{\mathcal{B}}$. The effects of the unobserved confounding and model misspecification can be balanced through introducing the coefficients $w_{u, i}$ for each $(u, i)$, by making 
\begin{equation}\label{balan}
\bfE[\boldsymbol{w}^{T}\mathcal{L}_{\mathcal{B}}(\boldsymbol{
x}; \theta, \phi)]=\bfE[\mathcal{L}_{\mathcal{U}}\left(\theta(\phi) ; \mathcal{U}\right)]=\bfE[\cL_{ideal}(\theta)].    
\end{equation}
Proposition \ref{th3} is the empirical version of Eq. (\ref{balan}) in terms of the balanced IPS, DR, and AutoDebias loss.
\begin{proposition}\label{th3}
(a) There exsits $w_{u, i}>0, (u, i)\in\mathcal{B}$ such that
\[
\sum_{(u, i) \in \mathcal{B}} w_{u, i}    \frac{ e_{u,i} }{ \hat p_{u, i} }=\frac{1}{|\cU|} \sum_{(u, i) \in \mathcal{U}}  e_{u, i}.\]
(b) There exsits $w_{u, i, 1}> 0, (u, i)\in\mathcal{D}$ and $w_{u, i, 2}> 0, (u, i)\in\mathcal{B}$ such that
\[\sum_{(u, i) \in \mathcal{D}}  w_{u, i, 1} \hat e_{u,i}  +  \sum_{(u, i) \in \mathcal{B}} w_{u, i, 2}    \frac{ e_{u,i}-\hat e_{u, i} }{ \hat p_{u, i} } =\frac{1}{|\cU|} \sum_{(u, i) \in \mathcal{U}}  e_{u, i}.\]
(c) There exsits $w_{u, i, 1}> 0, (u, i)\in\mathcal{D}$ and $w_{u, i, 2}> 0, (u, i)\in\mathcal{B}$ such that
\[\sum_{(u, i) \in \mathcal{D}}  w_{u, i, 1} \frac{\hat e_{u,i}}{\hat p_{u, i, 1}}  +  \sum_{(u, i) \in \mathcal{B}} w_{u, i, 2}    \frac{ e_{u,i} }{ \hat p_{u, i, 2} } =\frac{1}{|\cU|} \sum_{(u, i) \in \mathcal{U}}  e_{u, i}.\]
\end{proposition}
From Proposition \ref{th3}(a), when $w_{u, i}\equiv{|\cD|}^{-1}$, the left-hand side (LFS) degenerates to the standard IPS with maximal entropy of the balancing weights. The training objectives of the balanced IPS are
\begin{align}
\max_{\boldsymbol{w}\in  \mathbb{R}^{|\mathcal{B}|}}& \sum_{(u, i) \in \mathcal{B}} w_{u, i}  \log (w_{u, i})\label{eq8} \\
\text { s.t. } & ~
 w_{u, i}  > 0,\quad (u, i)\in \mathcal{B} \label{eq9} \\
 &\frac{1}{|\mathcal{B}|}\sum_{(u, i) \in \mathcal{B}} w_{u, i}=\frac{1}{|\mathcal{D}|} \label{eq10} \\
&\sum_{(u, i) \in \mathcal{B}} w_{u, i}    \frac{ e_{u,i} }{ \hat p_{u, i} }=\frac{1}{|\cU|} \sum_{(u, i) \in \mathcal{U}}  e_{u, i}, \label{eq11}
\end{align}
where the training objective in Eq. (\ref{eq8}) is to maximize the empirical entropy of the balancing weights and to be able to prevent extreme weights. The positivity and normality of the balancing weights are guaranteed by Eq. (\ref{eq9}) and Eq. (\ref{eq10}), respectively, and the influence of unobserved confounding and model misspecification is balanced out by reweighting the IPS estimates on biased ratings in Eq. (\ref{eq11}).

Similarly, for balanced DR and AutoDebias in Proposition \ref{th3}(b) and \ref{th3}(c), the estimators are re-weighted by $w_{u, i, 1}$ and $w_{u, i, 2}$ on the entire and biased user-item pairs, respectively, to combat unobserved confounding and model misspecification. The training objectives of the balanced DR are
\begin{align}
\max_{\boldsymbol{w}_1, \boldsymbol{w}_2}& \sum_{(u, i) \in \mathcal{D}} w_{u, i, 1}  \log (w_{u, i, 1})+ \sum_{(u, i) \in \mathcal{B}} w_{u, i, 2}  \log (w_{u, i, 2})\label{obj2}\\
\text { s.t. } & ~
 w_{u, i, 1}  > 0,\quad (u, i)\in \mathcal{D},\qquad w_{u, i, 2}  > 0,\quad (u, i)\in \mathcal{B}\label{eq13} \\
 &\sum_{(u, i) \in \mathcal{D}} w_{u, i, 1}=1, \quad \frac{1}{|\mathcal{B}|}\sum_{(u, i) \in \mathcal{B}} w_{u, i, 2}=\frac{1}{|\mathcal{D}|}\label{eq14} \\
 &{\small\sum_{(u, i) \in \mathcal{D}}  w_{u, i, 1} \hat e_{u,i}  +  \sum_{(u, i) \in \mathcal{B}} w_{u, i, 2}    \frac{ e_{u,i}-\hat e_{u, i} }{ \hat p_{u, i} } =\frac{1}{|\cU|} \sum_{(u, i) \in \mathcal{U}}  e_{u, i},} \label{eq15}
\end{align}
where $\boldsymbol{w}_1=[w_{u, i, 1}\mid (u, i)\in \cD]$, $\boldsymbol{w}_2=[w_{u, i, 2}\mid (u, i)\in \mathcal{B}]$, and the difference in balanced AutoDebias is that Eq. (\ref{eq15}) comes to
\begin{align}
{\small\sum_{(u, i) \in \mathcal{D}}  w_{u, i, 1} \frac{\hat e_{u,i}}{\hat p_{u, i, 1}}  +  \sum_{(u, i) \in \mathcal{B}} w_{u, i, 2}    \frac{ e_{u,i} }{ \hat p_{u, i, 2} } =\frac{1}{|\cU|} \sum_{(u, i) \in \mathcal{U}}  e_{u, i},} \label{eq16}
\end{align}
where the LFS of Eq. (\ref{eq15}) and Eq. (\ref{eq16}) degenetates to standard DR and AutoDebias, respectively, when $w_{u, i, 1}\equiv{|\cD|}^{-1}$ on $\cD$ and $w_{u, i, 1}\equiv{|\cD|}^{-1}$ on $\mathcal{B}$. Theorem \ref{the4} proves the existence of global optimal solutions corresponding to the proposed balanced IPS, DR and AutoDebias using Karush-Kuhn-Tucker conditions.
\begin{theorem}
There exists global optimal solutions to the optimization problem in balanced IPS, DR and AutoDebias. \label{the4}
\end{theorem}
\begin{proof}
Note that the empirical entropy as the optimization objectives in Eq. (\ref{eq8}) and Eq. (\ref{obj2}) are strictly convex. The inequality constraints in Eq. (\ref{eq9}) and Eq. (\ref{eq13}) are strictly feasible, i.e., there exists $w_{u, i}$ in $\cD$ such that $w_{u, i} > 0$. The equality constraints are affine in Eq. (\ref{eq10}), Eq. (\ref{eq11}), Eq. (\ref{eq14}), and Eq. (\ref{eq15}). By the Karush-Kuhn-Tucker condition, there exist global optimal solutions.
\end{proof}
Theoretically, due to the convexity of the objective function, its local optimal solution is same as the global optimal solution. The generalized Lagrange multiplier method can be used to solve the primal and the dual problem, and such balancing weights can effectively combat the unobserved confounding as in Proposition \ref{th3}.

\subsection{Training Algorithm}\label{sec:3.3}
Next, we propose an efficient mode-agnostic training algorithm to achieve the training objective in Section \ref{sec:3.2}. The algorithm consists of three parts: first, training the propensity and imputation models using a bi-level optimization, \emph{but without updating the prediction model;} then, reparameterizing and updating the gradients of the balancing weights to combat the effects of unobserved confounding and model misspecification; and finally, \emph{minimizing the estimated balancing loss,} named Bal-IPS, Bal-DR, or Bal-AutoDebias, and updating the prediction model to achieve unbiased learning.

\subsubsection{Propensity and Imputation Model Training}
Different from LTD and AutoDebias that use bi-level optimization to update the prediction model, we only perform assumed updates of the prediction model parameters $\theta(\phi)$ using bi-level optimization by Eq. (\ref{eq6}), and updates of the propensity and imputation model parameters $\phi$ by Eq. (\ref{eq-5}). Since there may exist unobserved confounding or model misspecification, we postpone the true update of the prediction model parameters $\theta$ to Section 3.3.3, after performing the balancing steps  in Section 3.3.2. We summarized the propensity and imputation model training algorithm in Alg. \ref{alg1}.
\begin{algorithm}[t]
\caption{\small Propensity and Imputation Model Training}
\label{alg1}
\LinesNumbered
\KwIn{$S$, $\mathbf{R}^\mathcal{B}$, $\mathbf{R}^\mathcal{U}$, $\phi_0$, $\theta_0$, $\eta$}
\For{$s = 0, \dots, S-1$}{
Sample mini-batches $\mathcal{B}_s \subseteq \mathcal{B}$ and $\mathcal{U}_s \subseteq \mathcal{U}$\;
Compute the lower loss in Eq. (\ref{eq6}) on $\mathcal{B}_s$\;
Compute an assumed update  $\theta_{s+1}(\phi_s)\leftarrow \theta_s-\eta\nabla_{\theta_s}\mathcal{L}_{\mathcal{B}}(\theta, \phi ; \mathcal{B}_s)$\;
Compute the upper loss in Eq. (\ref{eq-5}) on $\mathcal{U}_s$\;
Update the propensity and imputation model $\phi_{s+1}\leftarrow \phi_s-\eta\nabla_{\phi_s}\mathcal{L}_{\mathcal{U}}(\theta_{s+1}(\phi) ; \mathcal{B}_s)$\;
}
\KwOut{$\phi_S$}
\end{algorithm}

\subsubsection{Balancing Unobserved Confounding Training}
One challenge in solving the balancing optimization problem is that as the number of user-item pairs increases, the number of balancing weights also increases, resulting in a significant increase in solution time for large-scale datasets. To address this issue, we propose to \emph{reparameterize} $w_{u, i}$ in the balanced IPS, i.e., $w_{u, i} = g(x_{u, i}; \xi)$, where $\xi$ is the balancing model parameter. To satisfy the optimization constraints Eq. (\ref{eq9}) and Eq. (\ref{eq10}), the last layer of $g(x_{u, i}; \xi)$ uses $\textsc{Sigmoid}$ as the activation function to guarantee positivity and batch normalization to guarantee normality. The balancing weights in the balanced IPS are trained by minimizing the negative empirical entropy with the violation of the balanced constraint Eq. (\ref{eq11}) as regularization
{\begin{align}
\cL_{W-IPS}(\xi)=&-\sum_{(u, i) \in \mathcal{B}} w_{u, i}  \log (w_{u, i})\notag\\
&+{} \lambda \Bigg[ \sum_{(u, i) \in \mathcal{B}}  w_{u, i} \frac{e_{u,i}}{\hat p_{u, i}}  -\frac{1}{|\cU|} \sum_{(u, i) \in \mathcal{U}}  e_{u, i}\Bigg ]^2,\notag
\end{align}}where $\lambda>0$ is a hyper-parameter, for trade-off the original loss estimation with the correction due to the unobserved confounding. 

Similarly, $w_{u, i,1}$ and $w_{u, i,2}$ in the balanced DR and balanced AutoDebias are also reparameterized as $w_{u, i, 1}=g(x_{u, i}; \xi_1)$ and $w_{u, i, 2}=g(x_{u, i}; \xi_2)$. The balancing weights in the balanced DR and balanced AutoDebias are trained by minimizing
{\small
\begin{align}
\cL&_{W-DR}(\xi)=-\sum_{(u, i) \in \mathcal{D}} w_{u, i, 1}  \log (w_{u, i, 1})-\sum_{(u, i) \in \mathcal{B}} w_{u, i, 2}  \log (w_{u, i, 2})\notag\\
&+{} \lambda \Bigg[ \sum_{(u, i) \in \mathcal{D}}  w_{u, i, 1} \hat e_{u,i}  +\sum_{(u, i) \in \mathcal{B}} w_{u, i, 2}    \frac{ e_{u,i}-\hat e_{u, i} }{ \hat p_{u, i} } -\frac{1}{|\cU|} \sum_{(u, i) \in \mathcal{U}}  e_{u, i}\Bigg]^2,\notag
\end{align}}
and
{\small
\begin{align}
\cL&_{W-Auto}(\xi)=-\sum_{(u, i) \in \mathcal{D}} w_{u, i, 1}  \log (w_{u, i, 1})-\sum_{(u, i) \in \mathcal{B}} w_{u, i, 2}  \log (w_{u, i, 2})\notag\\
&+{} \lambda \Bigg[ \sum_{(u, i) \in \mathcal{D}}  w_{u, i, 1} \frac{\hat e_{u,i}}{\hat p_{u, i, 1}}  +\sum_{(u, i) \in \mathcal{B}} w_{u, i, 2}    \frac{ e_{u,i}}{ \hat p_{u, i, 2} } -\frac{1}{|\cU|} \sum_{(u, i) \in \mathcal{U}}  e_{u, i}\Bigg ]^2,\notag
\end{align}}where $\lambda>0$ is a hyper-parameter, and $\xi\equiv\{\xi_1, \xi_2\}$ are the parameters of the balancing model.

\subsubsection{Prediction Model Training}
Since the optimization of the balancing weights aims to balance the prediction errors on the biased and unbiased ratings, which also depends on the prediction model, we propose to update the balancing model and the prediction model alternatively. Specifically, given the balancing weights of IPS, the prediction model is trained by minimizing the balanced IPS (Bal-IPS)
{\begin{align}
\cL_{Bal-IPS}(\theta)=\sum_{(u, i) \in \mathcal{B}}  w_{u, i} \frac{e_{u,i}}{\hat p_{u, i}}.
\end{align}}

Similarly, for balanced DR (Bal-DR) or balanced AutoDebias (Bal-AutoDebias), the prediction model is trained by minimizing
\begin{align}
\cL&_{Bal-DR}(\theta)=\sum_{(u, i) \in \mathcal{D}}  w_{u, i, 1} \hat e_{u,i}  +\sum_{(u, i) \in \mathcal{B}} w_{u, i, 2}    \frac{ e_{u,i}-\hat e_{u, i} }{ \hat p_{u, i}},
\end{align}
and
\begin{align}
\cL&_{Bal-Auto}(\theta)=\sum_{(u, i) \in \mathcal{D}}  w_{u, i, 1} \frac{\hat e_{u,i}}{\hat p_{u, i, 1}}  +\sum_{(u, i) \in \mathcal{B}} w_{u, i, 2}    \frac{ e_{u,i}}{ \hat p_{u, i, 2}}.
\end{align}
Next, given the prediction model, the balancing weights are updated again as described in Section 3.3.2. The balancing weights and the prediction model are updated alternately, allowing a more adequate use of unbiased ratings, resulting in unbiased learning of the prediction model.

The main difference compared with LTD~\cite{Wang-etal2021} and AutoDebias~\cite{Chen-etal2021} is that we do not only use unbiased ratings to select the parameters of the propensity and imputation models, and then use standard IPS or DR for the prediction model update. Instead, we combat the effects of unobserved confounding by introducing a balancing model, and then perform prediction model updates based on the balanced losses. Remarkably, the proposed method is model-agnostic and can be applied to any of the debiased recommendation methods. Here we use IPS, DR and AutoDebias for illustration.  We summarized the whole training algorithm in Alg. \ref{alg2}.
\begin{algorithm}[t]
\caption{\small Balancing Unobserved Confounding Training}
\label{alg2}
\LinesNumbered
\KwIn{$T$, $S$, $\mathbf{R}^\mathcal{B}$, $\mathbf{R}^\mathcal{U}$, $\phi_0$, $\theta_0$, $\xi_0$, $\eta$, $\lambda$}
\For{$t = 0, \dots, T-1$}{
Call Alg. 1 by $\phi_{t+1}\leftarrow\text{Alg. 1}(S, \mathbf{R}^\mathcal{B}, \mathbf{R}^\mathcal{U}, \phi_t, \theta_t, \eta)$\;
\For{$s = 0, \dots, S-1$}{
Sample mini-batches $\mathcal{D}^s_t \subseteq \mathcal{D}$, $\mathcal{B}^s_t \subseteq \mathcal{B}$ and $\mathcal{U}^s_t \subseteq \mathcal{U}$\;
Compute unmeasured confounding balancing loss\;
Update the balancing weight $\xi_t^{s+1}\leftarrow\xi_t^{s}-\eta\nabla_{\xi_t^{s}}\cL_{W}(\xi)$\;
Compute the balanced prediction error loss\;
Update the prediction model $\theta_t^{s+1}\leftarrow\theta_t^{s}-\eta\nabla_{\theta_t^{s}}\cL_{Bal}(\theta)$\;
}
Copy the balancing model’s parameters $\xi_{t+1}^0 \leftarrow \xi_t^S$\;
Copy the prediction model’s parameters $\theta_{t+1}^0 \leftarrow \theta_t^S$\;
}
\KwOut{$\theta_T$}
\end{algorithm}
 
\subsubsection{Training Efficiency}
\begin{figure}[t]
    \centering
    \includegraphics[scale=0.42]{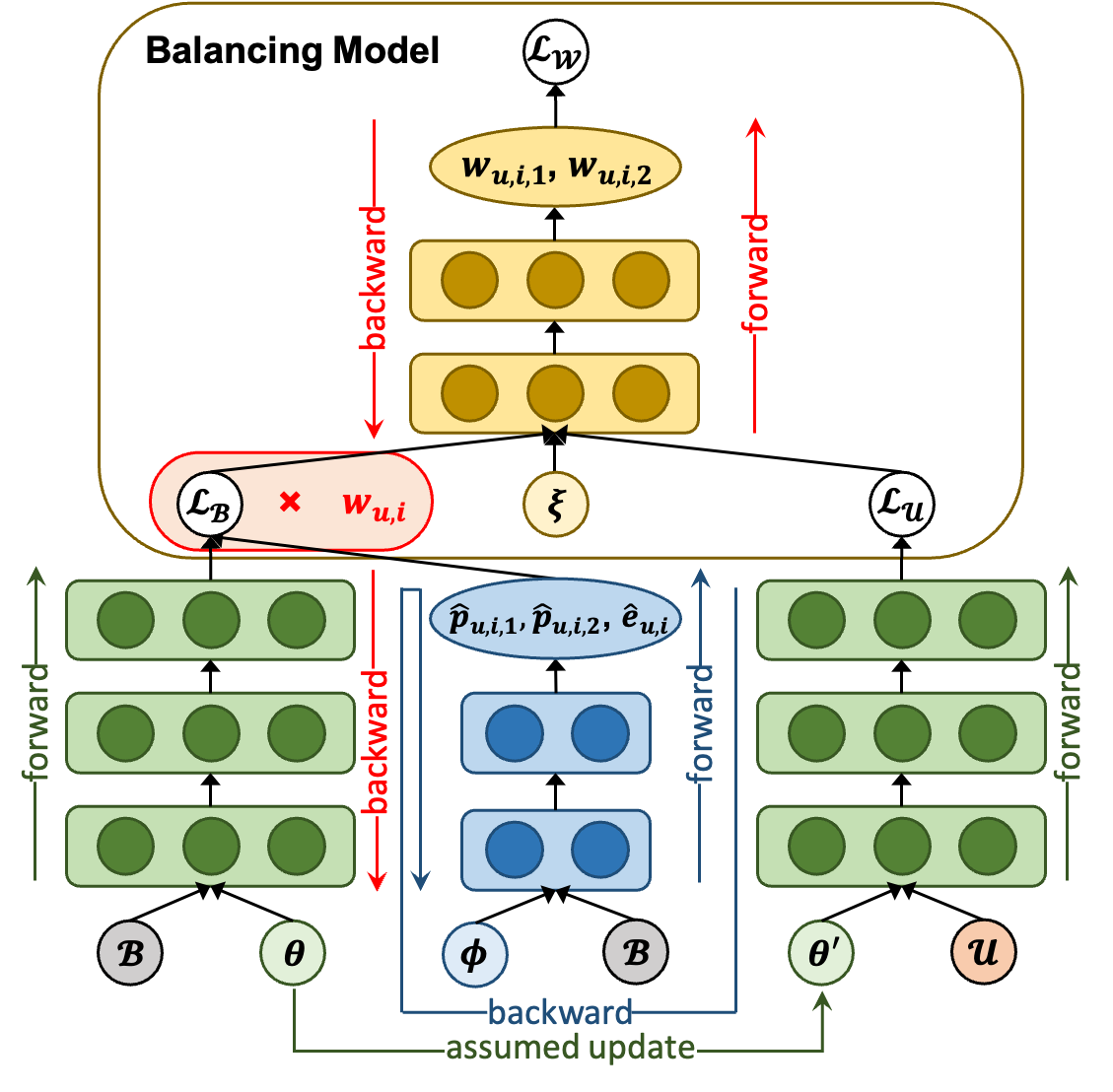}
    \vspace{-12pt}
    \caption{The proposed workflow for balancing unobserved confounding consists of four steps: (1) assumed updating the prediction model parameters from $\theta(\phi)$ to $\theta^{\prime}(\phi)$ using $\mathbf{R}^\mathcal{B}$ (green arrow); (2) updating the propensity and imputation model parameters $\phi$ using $\mathbf{R}^\mathcal{U}$ (blue arrow); (3) updating the balancing model parameters $\phi$ using both $\mathbf{R}^\mathcal{B}$ and $\mathbf{R}^\mathcal{U}$ (red arrow); (4) actually updating the prediction model parameters $\theta$ using the balanced loss $\boldsymbol{w}^{T}\mathcal{L}_{\mathcal{B}}$ (red arrow).}
    \label{fig:share} 
    \vspace{-6pt}
\end{figure}
\begin{table}[t]
 \centering
  \small
 \setlength{\tabcolsep}{4.7pt}
\captionof{table}{Summary of the datasets.}\label{data}
\vspace{-0.3cm} 
\begin{tabular}{ccccccc}
\toprule
             &  Users& Items &Training &Uniform &Validation &Test\\ \midrule
\textsc{Music}       & 15,400 & 1,000  & 311,704  & 2,700  & 2,700 & 48,600\\
\textsc{Coat}       & 290 & 300 & 6,960 & 232 & 232 & 4,176\\ \bottomrule
\end{tabular}
\vskip -0.1in
\end{table}
The proposed workflow for balancing the unobserved confounding is shown in Figure \ref{fig:share}. In Section 3.3.1, our algorithm performs two forward and backward passes for the prediction model on $\mathbf{R}^\mathcal{B}$ and $\mathbf{R}^\mathcal{U}$, respectively, and one forward and backward pass for the propensity and imputation model on $\mathbf{R}^\mathcal{B}$. The backward-on-backward pass is used to obtain the gradients of the propensity and imputation models. In Section 3.3.2, one forward and one reverse pass are performed for the balancing model. In Section 3.3.3, a backward pass is used to actually update the prediction model. We refer to~\cite{Ren-etal2018, Wang-etal2021} that the running time of a backward-on-backward pass and a forward pass are about the same. As a result, the training time of the proposed algorithm does not exceed 3x learning time compared to two-stage learning and about 1.5x learning time compared to LTD and AutoDebias.
\section{Real-world Experiments}
\begin{table*}[]  
\centering
\setlength{\tabcolsep}{4pt}
\caption{Performance comparison in terms of AUC, NDCG@5, and NDCG@10. The best results to each base method are bolded.}

\vspace{-10pt}
\label{tab:performance}
\begin{tabular}{c|cccccc|cccccc}
\toprule
\multirow{2}{*}{Method}          & \multicolumn{6}{c|}{\textsc{Music}}                                                       & \multicolumn{6}{c}{\textsc{Coat}}                                                                              \\
\cmidrule(lr){2-7}
\cmidrule(lr){8-13}
& AUC                             & RI    & NDCG@5                           & RI    & NDCG@10                             & RI    & AUC                           & RI    & NDCG@5                             & RI    & NDCG@10                          & RI     \\ \midrule
CausE        & 0.731                           & -     & 0.551                         & -     & 0.656                          & -     & 0.761                           & -     & 0.500                           & -     & 0.605                           & -      \\\midrule
KD-Label               & 0.740                          & -     &  0.580                         & -     & 0.680                         & -     & 0.750                           & -     &  0.504                           & -     & 0.610                    & -      \\ \midrule
MF (biased)               & 0.727                           & -     & 0.550                           & -     & 0.655                           & -     & 0.747                           & -     & 0.500                           & -     & 0.606                           & -      \\
MF (uniform)          & 0.573                    & -     & 0.449         & -           & 0.591                   & -                      & 0.579    & -     & 0.358 & -                           & 0.482                        & -   \\
MF (combine)               & 0.730                           & -     & 0.554                           & -     & 0.659                           & -     & 0.750                           & -     & 0.503                      & -     & 0.611                           & - 
 \\
Bal-MF            & \textbf{0.739}                          & 1.23\%  & \textbf{0.579}              & 4.51\%   & \textbf{0.679}  & 3.03\% & \textbf{0.761}  & 1.47\% & \textbf{0.511}                          & 1.59\%       & \textbf{0.620}          & 1.47\% 
\\ \midrule
IPS                & 0.723                           & -     & 0.549                           & -     & 0.656                    & -     & 0.760                           & -     & 0.509                           & -     & 0.613                           & -      \\
Bal-IPS            & \textbf{0.727}                           & 0.55\% & \textbf{0.564}                           & 2.73\% & \textbf{0.668} & 1.83\% & \textbf{0.771} & 1.45\% & \textbf{0.521}                           & 2.36\% & \textbf{0.628}                           & 2.45\% \\ \midrule
DR                & 0.724                           & -     & 0.550                           & -     & 0.656                           & -     & 0.765                           & -     & 0.521                           & -     & 0.620                           & -      \\
Bal-DR            & \textbf{0.731}                           & 0.97\% &\textbf{0.569}                           & 3.45\% & \textbf{0.669} & 1.98\% & \textbf{0.770} & 0.65\% & \textbf{0.523}                           & 0.38\% & \textbf{0.628}                           & 1.29\% \\ \midrule
AutoDebias        & 0.741                           & -     & 0.645                           & -     & 0.725                           & -     & 0.766                            & -     & 0.522                          & -     & 0.621                           & -      \\
Bal-AutoDebias    & \textbf{0.749}                           & 1.08\% & \textbf{0.670}                           & 3.88\% & \textbf{0.744}                           & 2.62\% & \textbf{0.772}                           & 0.78\% & \textbf{0.544}                           & 4.21\% & \textbf{0.640} & 3.06\% \\ \bottomrule
\end{tabular}
\begin{tablenotes}
\footnotesize
\item {Note: RI refers to the relative improvement of Bal-methods over the corresponding baseline.}
\end{tablenotes}
\vspace{-8pt}
\end{table*}
In this section, We conduct extensive experiments on two real-world datasets to answer the following research questions (RQs):
\begin{enumerate}
 \item[\bf RQ1.]  Do the proposed Bal-methods improve the debiasing performance compared with the existing methods?
 \item[\bf RQ2.]  Do our methods stably perform well with different initializations of the prediction model?
 \item[\bf RQ3.]  How does the balancing model affect the performance of our methods?
 \item[\bf RQ4.] What factors influence the effectiveness of our methods?
    \end{enumerate}

\subsection{Experimental Setup}
{\bf Dataset and preprocessing.} Following the previous studies~\citep{Wang-Zhang-Sun-Qi2019, saito2020doubly,Wang-etal2021, Chen-etal2021}, we conduct extensive experiments on the two widely used real-world datasets with both missing-not-at-random (MNAR) and missing-at-random (MAR) ratings: {\bf \textsc{{Music}}\footnote{http://webscope.sandbox.yahoo.com/}} and {\bf \textsc{{Coat}}\footnote{https://www.cs.cornell.edu/\textasciitilde schnabts/mnar/}}. 
In particular, \textsc{Music} dataset contains 15,400 users and 1,000 items with 54,000 MAR and 311,704 MNAR ratings. \textsc{Coat} dataset contains 290 users and 300 items with 4,640 MAR and 6,960 MNAR ratings.   Following~\citep{liu2020general,Chen-etal2021}, we take all the biased data as the training set and randomly split the uniform data as three parts: 5\% for balancing the unobserved confounding, 5\% for validation set and 90\% for test set. We summarize the datasets and splitting details in Table \ref{data}.
\\
{\bf Baselines.}
In our experiments, we compare the proposed Bal-methods with the following baselines: \\
$\bullet$ {\bf Base Model} ~\cite{koren2009matrix}: the Matrix Factorization (MF) model is trained on biased data, uniform data and both of them respectively, denoted as MF (biased), MF (uniform) and MF (combine).\\
$\bullet$ {\bf  Inverse Propensity Scoring (IPS)}~\cite{Schnabel-Swaminathan2016}: a reweighting method using inverse propensity scores to weight the observed events. \\
$\bullet$ {\bf  Doubly Robust (DR)}~\cite{Wang-Zhang-Sun-Qi2019, saito2020doubly}: an efficient method combining imputations and inverse propensities with double robustness. \\
$\bullet$ {\bf  CausE}~\cite{liu2020general}: a sample-based knowledge distillation approach to reduce computational complexity. \\
$\bullet$ {\bf  KD-Label}~\cite{liu2020general}: an efficient framework for knowledge distillation to transfer unbiased information to teacher model and guide the training of student model. \\
$\bullet$ {\bf  AutoDebias}~\cite{Chen-etal2021}: a meta-learning based method using few unbiased data to further mitigate the selection bias.  \\
{\bf Experimental protocols and details.} Following~\cite{Wang-etal2021, Chen-etal2021}, AUC, NDCG@5 and NDCG@10 are adopted as the evaluation metrics to measure the debiasing performance. Formally,
\[
A U C=\frac{\sum_{(u, i) \in \mathcal{U}^{+}} \hat{Z}_{u, i}-\left|\mathcal{U}^{+}\right|\cdot(\left|\mathcal{U}^{+}\right|+1) / 2}{\left|\mathcal{U}^{+}\right| \cdot(\left|\mathcal{U}\right|-\left|\mathcal{U}^{+}\right|)},
\]
and NDCG@k measures the quality of ranking list as
{\small
\[
D C G_u @ k=\sum_{(u, i) \in \mathcal{U}} \frac{\mathbb{I}(\hat{Z}_{u, i} \leq k)}{\log (\hat{Z}_{u, i}+1)},\quad
N D C G @ k=\frac{1}{M} \sum_{m=1}^M \frac{D C G_{u_m} @ k}{I D C G_{u_m}@ k},
\]}where $\mathcal{U}^{+}\subseteq \cU$ denotes the positive ratings in the uniform dataset, $\hat{Z}_{u, i}$ is the rank position of $(u, i)$ given by the rating predictions, and $ID C G_{u_m} @ k$ is the ideal $D C G_{u_m} @ k$.

All the methods are implemented on PyTorch. Throughout, \textsc{Adam} optimizer is utilized for propensity and imputation model with learning rate and weight decay in [1e-4, 1e-2]. \textsc{SGD} optimizer is utilized for prediction model and balancing model with learning rate in [1e-7, 1] and weight decay in [1e-4, 1]. We tune the regularization hyper-parameter $\lambda$ in \{$0, 2^{-9}$, $2^{-6}$, $2^{-3}$, 1\}. All hyper-parameters are tuned based on the performance on the validation set. 

\subsection{Performance Comparison (RQ1)}
\begin{table*}[]  
 \centering
 \setlength{\tabcolsep}{3pt}
 \captionof{table}{Performance of the Bal-methods under different prediction models as initializations on \textsc{Music} and \textsc{Coat}.}
\vspace{-10pt}
\label{tab3}
\begin{tabular}{c|c|ccc|ccc|ccc}
\toprule
        & Initial Method   & \multicolumn{3}{c|}{Initial with IPS} & \multicolumn{3}{c|}{Initial with DR} & \multicolumn{3}{c}{Initial with AutoDebias} \\ \midrule
Dataset & Method           & AUC     & NDCG@5    & NDCG@10    & AUC    & NDCG@5    & NDCG@10    & AUC       & NDCG@5      & NDCG@10      \\
\midrule
        & Baseline         & 0.723         & 0.549           & 0.656            & 0.724        & 0.550           & 0.656           & 0.741          & 0.645             & 0.725             \\
\multirow{2}{*}{\textsc{Music}}   & Bal-IPS        & ${0.726}_{0.4\% \uparrow}$       & $0.561_{2.2\% \uparrow}$          & $0.666_{1.5\% \uparrow}$           & $0.726_{0.3\% \uparrow}$        & $0.562_{2.2\% \uparrow}$           & $0.666_{1.5\% \uparrow}$           & $0.747_{0.8\% \uparrow}$          & $0.656_{1.7\% \uparrow}$            & $0.733_{1.1\% \uparrow}$            \\
        & Bal-DR         & $0.725_{0.3\% \uparrow}$        & $0.556_{1.3\% \uparrow}$          & $0.665_{1.4\% \uparrow}$           & $0.726_{0.3\% \uparrow}$       & $0.559_{1.6\% \uparrow}$          & $0.667_{1.7\% \uparrow}$           &  $0.748_{0.9\% \uparrow}$          & $0.658_{2.0\% \uparrow}$            & $0.734_{1.2\% \uparrow}$             \\
        & Bal-AutoDebias   &  $\textbf{0.739}_{2.2\% \uparrow}$         &  $\textbf{0.584}_{6.4\% \uparrow}$           &  $\textbf{0.683}_{4.1\% \uparrow}$           &  $\textbf{0.740}_{2.2\% \uparrow}$       &  $\textbf{0.586}_{6.5\% \uparrow}$          &  $\textbf{0.684}_{4.3\% \uparrow}$           &  $\textbf{0.749}_{1.1\% \uparrow}$          &  $\textbf{0.670}_{3.9\% \uparrow}$            &  $\textbf{0.744}_{2.6\% \uparrow}$      \\ \midrule
        & Baseline         & 0.760        & 0.509          & 0.613           & 0.765        & 0.521          & 0.620           & 0.766          & 0.522            & 0.621             \\
\multirow{2}{*}{\textsc{Coat}}    & Bal-IPS        &  $\textbf{0.771}_{1.4\% \uparrow}$          & $0.521_{2.4\% \uparrow}$        &  $0.628_{2.4\% \uparrow}$           &  $0.770_{0.7\% \uparrow}$       & $0.523_{0.4\% \uparrow}$          & $0.627_{1.1\% \uparrow}$           &  $0.770_{0.5\% \uparrow}$          & $0.523_{0.2\% \uparrow}$             & $0.629_{1.3\% \uparrow}$             \\
        & Bal-DR         & $0.770_{1.3\% \uparrow}$        & $0.523_{2.8\% \uparrow}$          &  $0.628_{2.4\% \uparrow}$          &  $0.771_{0.8\% \uparrow}$       & $0.522_{0.2\% \uparrow}$          &  $0.629_{1.5\% \uparrow}$          &  $0.770_{0.5\% \uparrow}$         &  $0.523_{0.2\% \uparrow}$           &  $0.629_{1.3\% \uparrow}$            \\
        & Bal-AutoDebias   & $\textbf{0.771}_{1.4\% \uparrow}$         &  $\textbf{0.531}_{4.3\% \uparrow}$         &  $\textbf{0.632}_{3.1\% \uparrow}$           & $\textbf{0.772}_{0.9\% \uparrow}$       &  $\textbf{0.539}_{3.5\% \uparrow}$          &  $\textbf{0.637}_{2.7\% \uparrow}$           & $\textbf{0.772}_{0.8\% \uparrow}$          &  $\textbf{0.544}_{4.2\% \uparrow}$            &  $\textbf{0.640}_{3.1\% \uparrow}$             \\ \bottomrule
\end{tabular}
\vspace{-6pt}
\end{table*}
Table \ref{tab:performance} compares the prediction performance of the various methods on two real-world datasets \textsc{Music} and \textsc{Coat}. We find that the proposed model-agnostic Bal-methods have significantly improved performance when applied to MF, IPS, DR and AutoDebias with respect to all metrics. Overall, Bal-AutoDebias exhibits the best performance. Impressively, although AutoDebias hardly improves the performance on \textsc{Coat} compared with DR as reported in \cite{Chen-etal2021}, the proposed Bal-AutoDebias improves 4.21\% and 3.06\% on NDCG@5 and NDCG@10 compared with the best baseline, respectively, validating the effectiveness of the proposed balancing approach.

In addition, MF using only uniform data exhibits the worst performance, due to its small size which causes unavoidable overfitting. Directly combining the biased and unbiased ratings increases the MF performance slightly and insignificantly. As in ~\cite{Chen-etal2021}, AutoDebias has the most competitive performance among the existing methods, due to the use of unbiased ratings for the parameter selection of the propensity and imputation models. However, as discussed in previous sections, the previous methods were unable to combat the potential unobserved confounding in the biased data. The proposed Bal-methods address this issue by further utilizing unbiased ratings to balance the loss estimates from biased ratings.


\begin{table}[t]
\centering
\small
\caption{Effects of balancing models on Bal-AutoDebias.}
\vspace{-10pt}
\setlength{\tabcolsep}{2pt}
\label{tab4}
\begin{tabular}{c|c|ccc|ccc}
\toprule
\multicolumn{2}{c|}{Method}          & \multicolumn{3}{c|}{\textsc{Music}}                                                       & \multicolumn{3}{c}{\textsc{Coat}}          \\ \midrule
$w_{u, i, 1}$ & $w_{u, i, 2}$  & AUC     & NDCG@5        & NDCG@10                   & AUC        & NDCG@5         & NDCG@10                
 \\ \midrule
MF & MF        & 0.749                  & 0.670                  & 0.744             & 0.772            & 0.544              & 0.640          \\
MF & NCF        & 0.745                  & 0.667                  & 0.742             & 0.769            & 0.539              & 0.635    \\
NCF & MF        & \textbf{0.762}                  & \textbf{0.675}                  & \textbf{0.748}             & \textbf{0.774}            & \textbf{0.548}              & \textbf{0.646}  \\
NCF & NCF        & 0.749                  & 0.671                  & 0.745             & 0.771            & 0.545              & 0.639  \\
\bottomrule
\end{tabular}
\vspace{-12pt}
\end{table}

\begin{figure}[t]
\centering
\subfigure[\textsc{Music}-AUC]{
\begin{minipage}[t]{0.47\linewidth}
\centering
\includegraphics[width=1\textwidth]{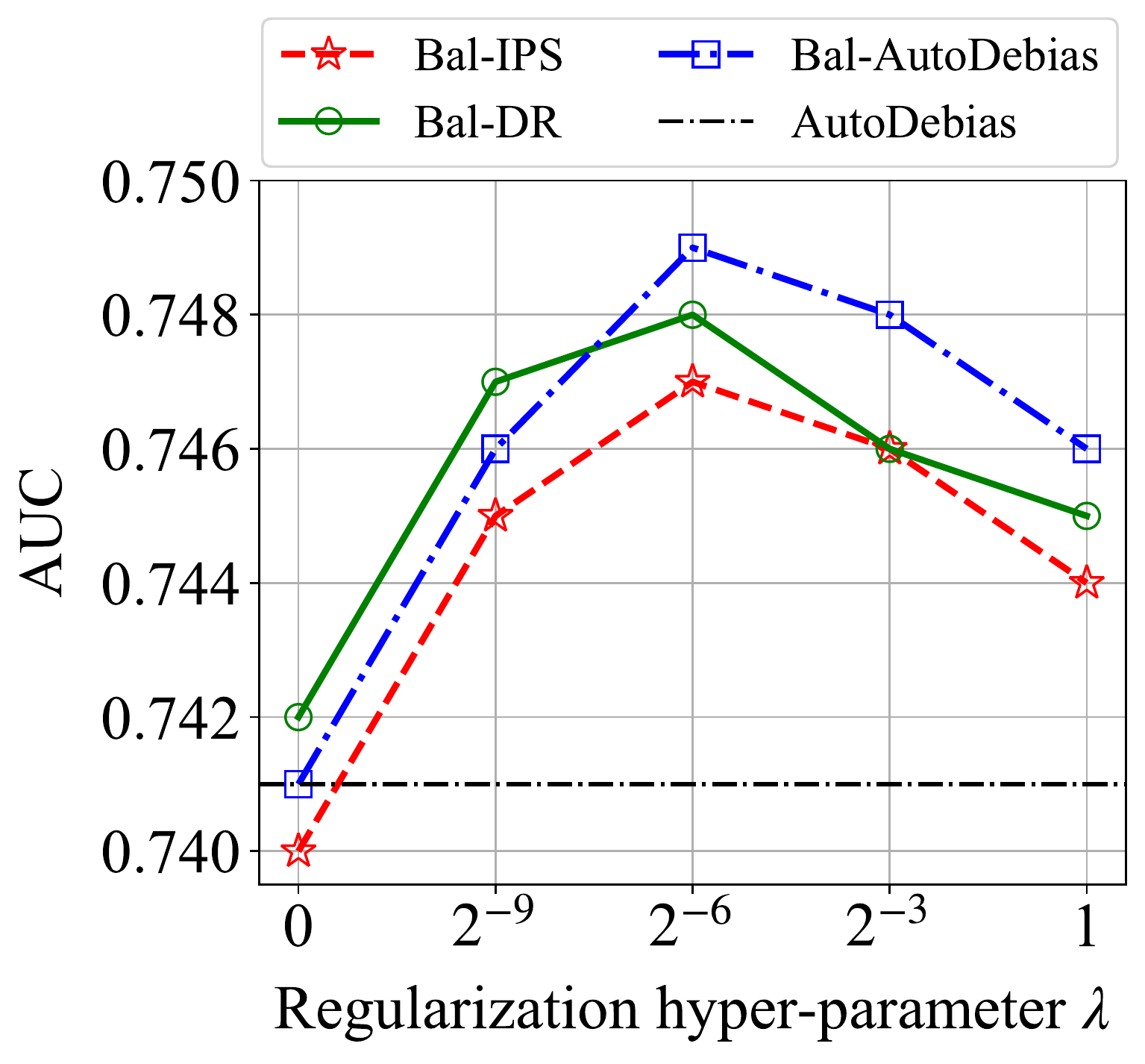}
\end{minipage}%
}%
\subfigure[\textsc{Coat}-AUC]{
\begin{minipage}[t]{0.47\linewidth}
\centering
\includegraphics[width=1\textwidth]{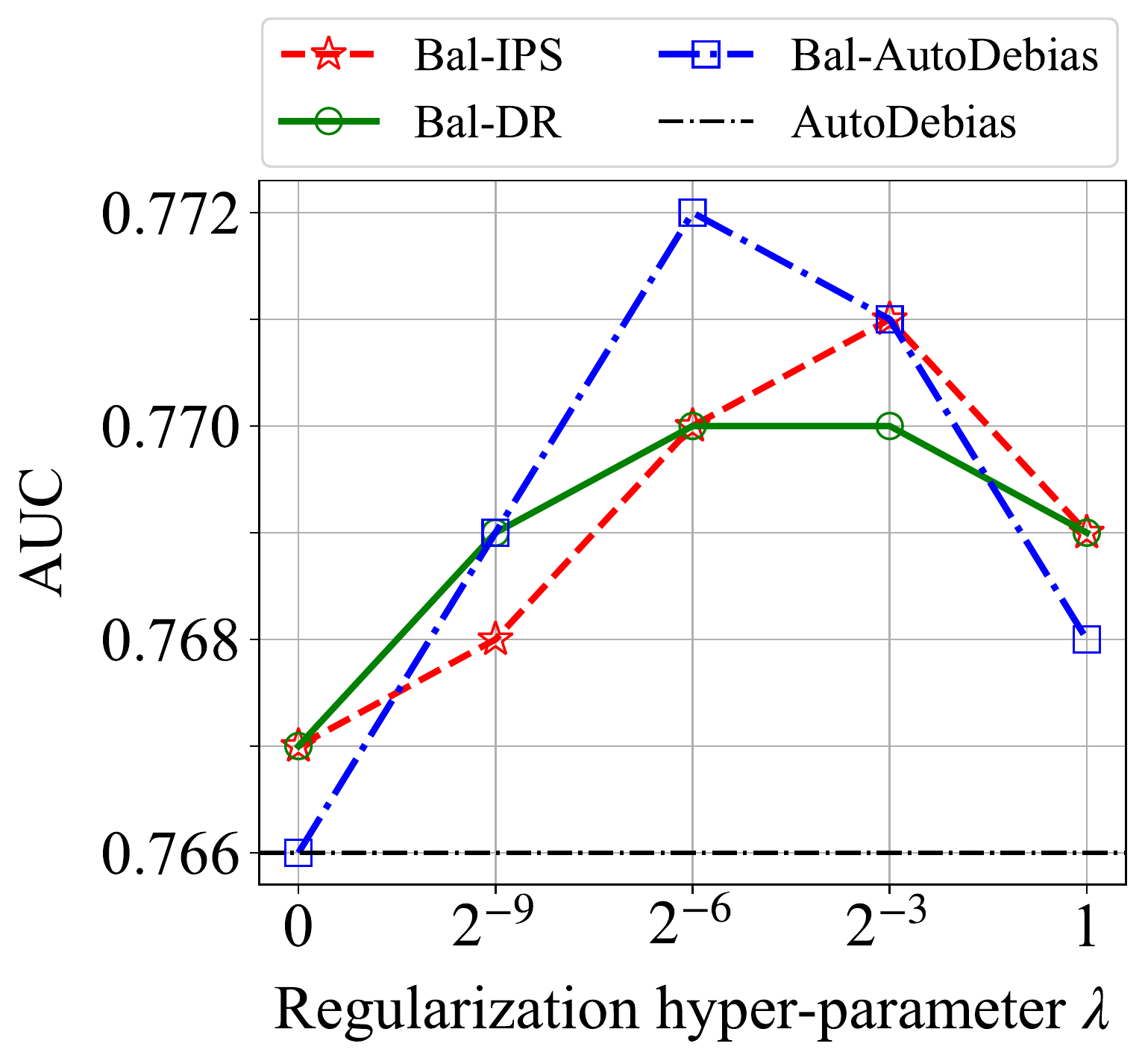}
\end{minipage}%
}%
\centering
\vspace{-8pt}

\centering
\subfigure[\textsc{Music}-NDCG@5]{
\begin{minipage}[t]{0.47\linewidth}
\centering
\includegraphics[width=1\textwidth]{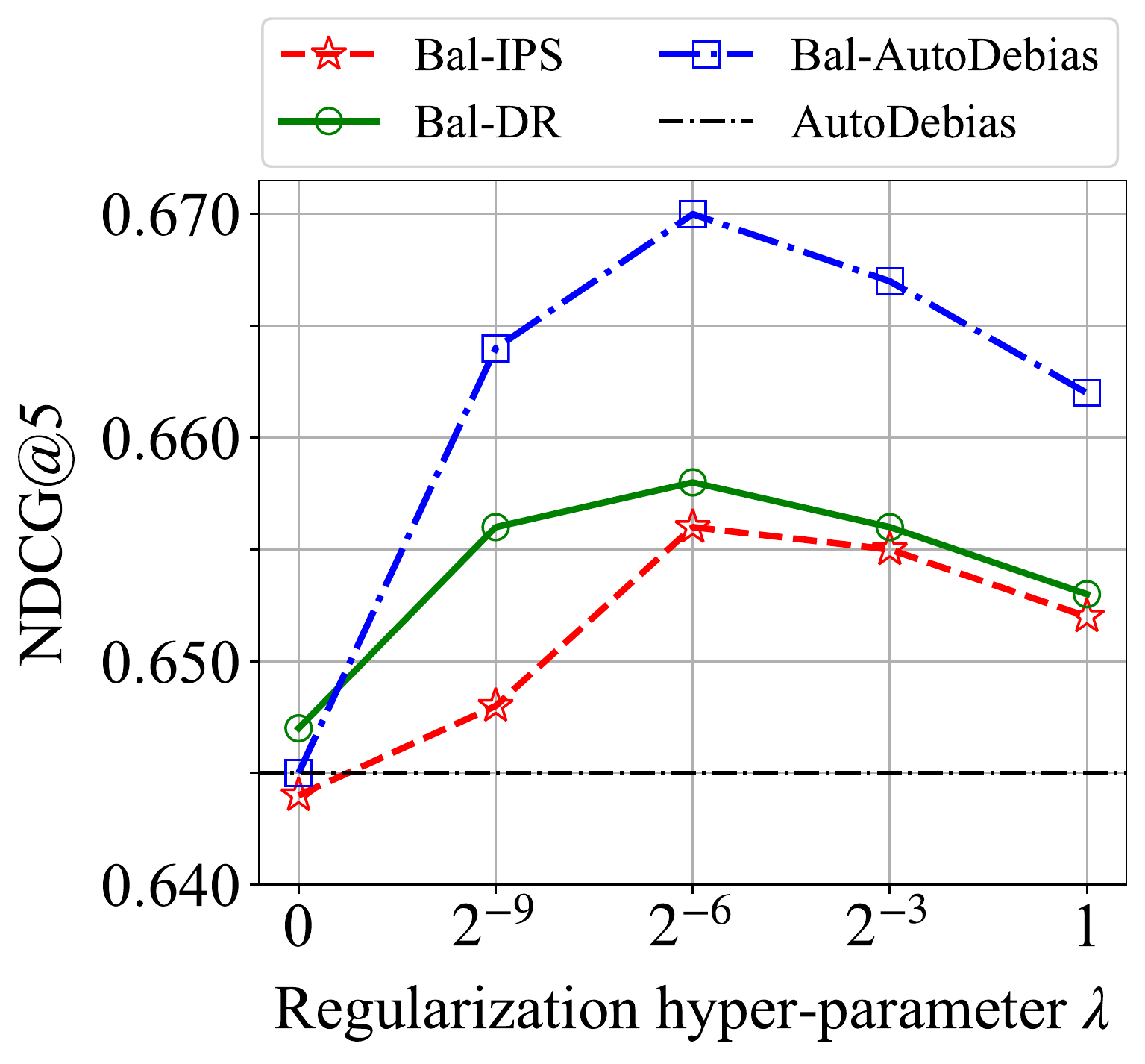}
\end{minipage}%
}
\subfigure[\textsc{Coat}-NDCG@5]{
\begin{minipage}[t]{0.47\linewidth}
\centering
\includegraphics[width=1\textwidth]{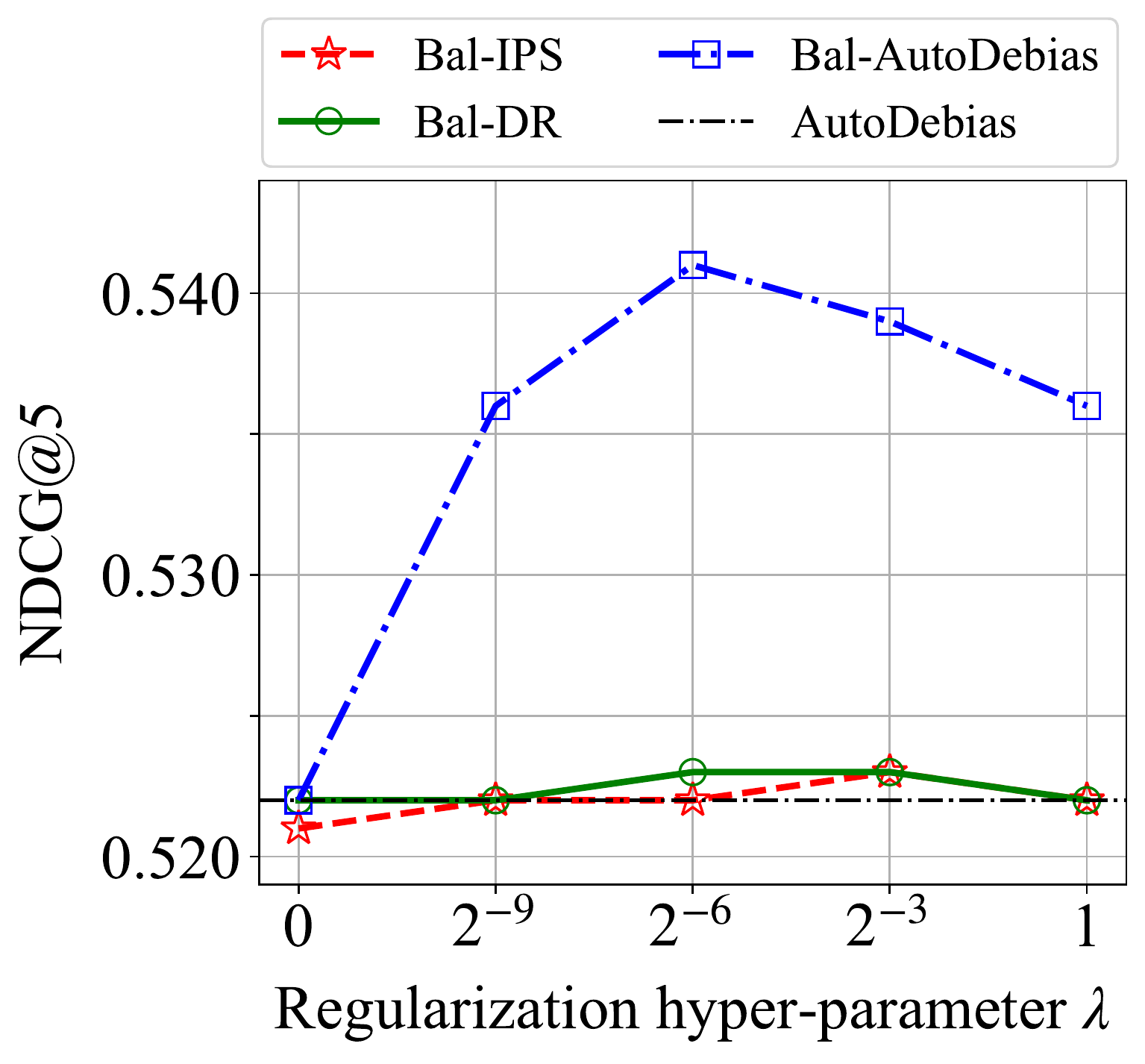}
\end{minipage}%
}%
\centering
\vspace{-8pt}

\centering
\subfigure[\textsc{Music}-NDCG@10]{
\begin{minipage}[t]{0.47\linewidth}
\centering
\includegraphics[width=1\textwidth]{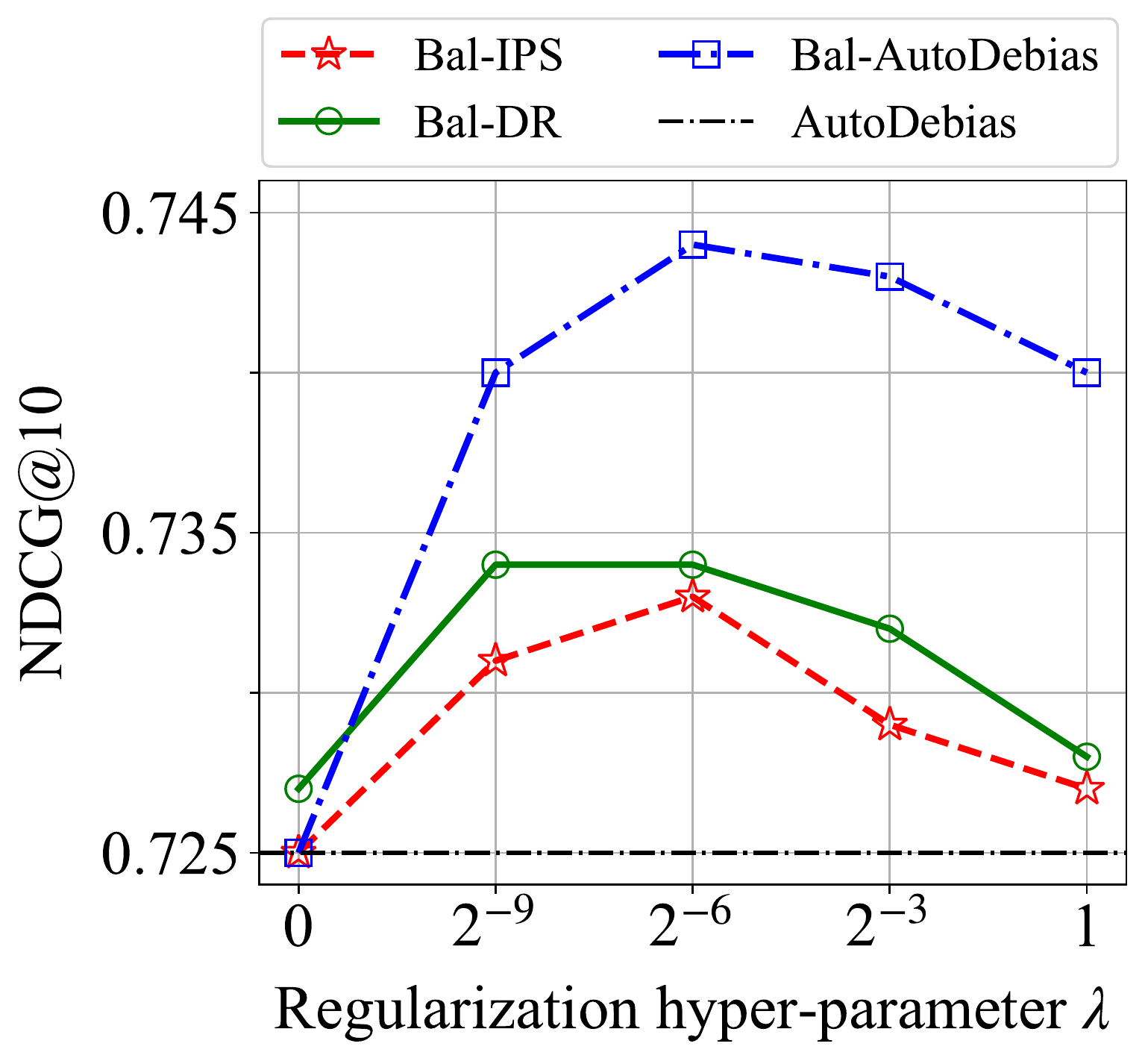}
\end{minipage}%
}%
\subfigure[\textsc{Coat}-NDCG@10]{
\begin{minipage}[t]{0.47\linewidth}
\centering
\includegraphics[width=1\textwidth]{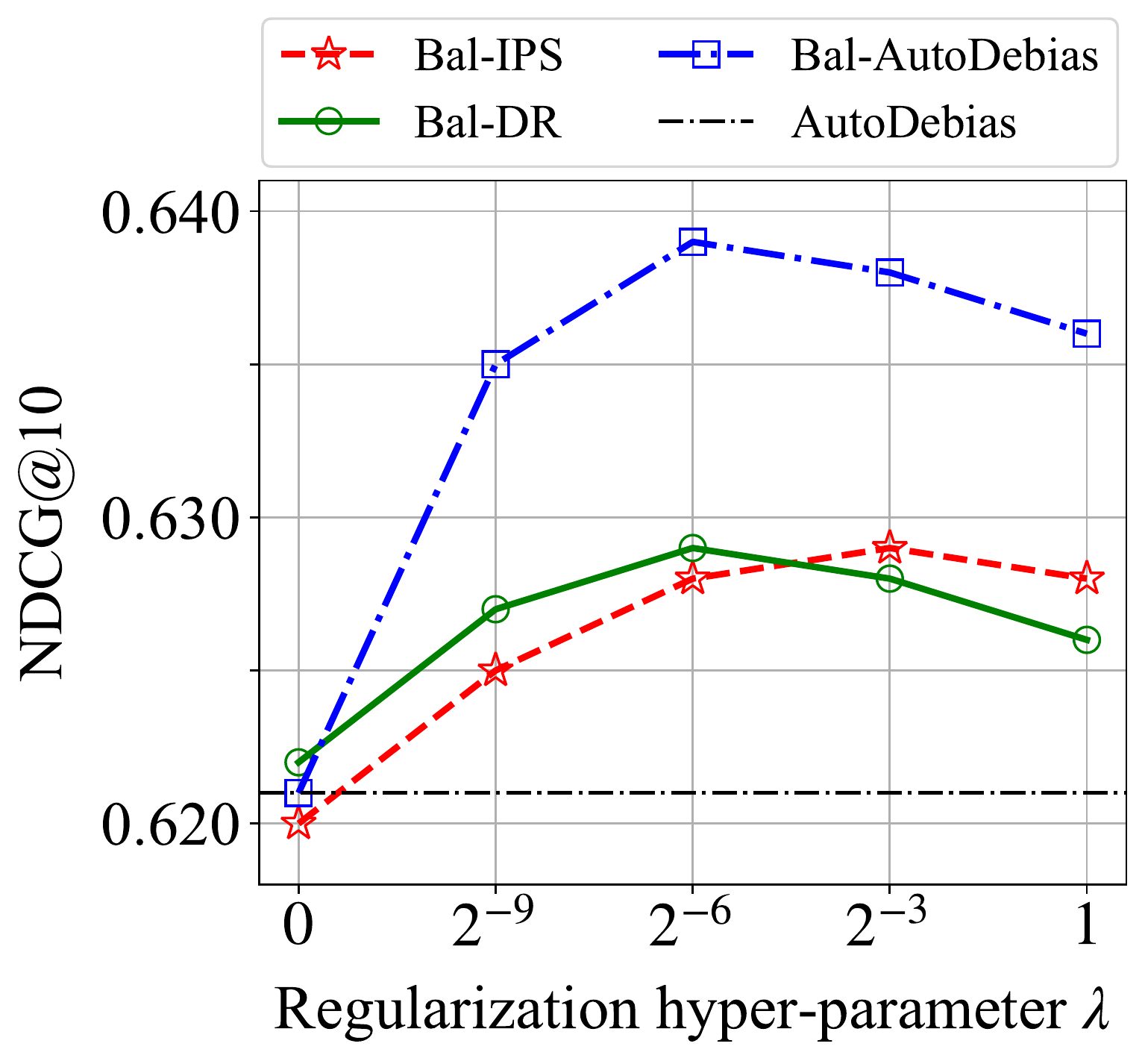}
\end{minipage}%
}%
\centering
\vspace{-8pt}
\caption{Effect of regularization strength $\lambda$ on \textsc{Music} and \textsc{Coat}, degenerating to standard AutoDebias when $\lambda=0$.}
\label{fig:confounder}
\vspace{-18pt}
\end{figure}

\begin{figure}[t]
\centering
\subfigure[\textsc{Music}-NDCG@5]{
\begin{minipage}[t]{0.495\linewidth}
\centering
\includegraphics[width=1\textwidth]{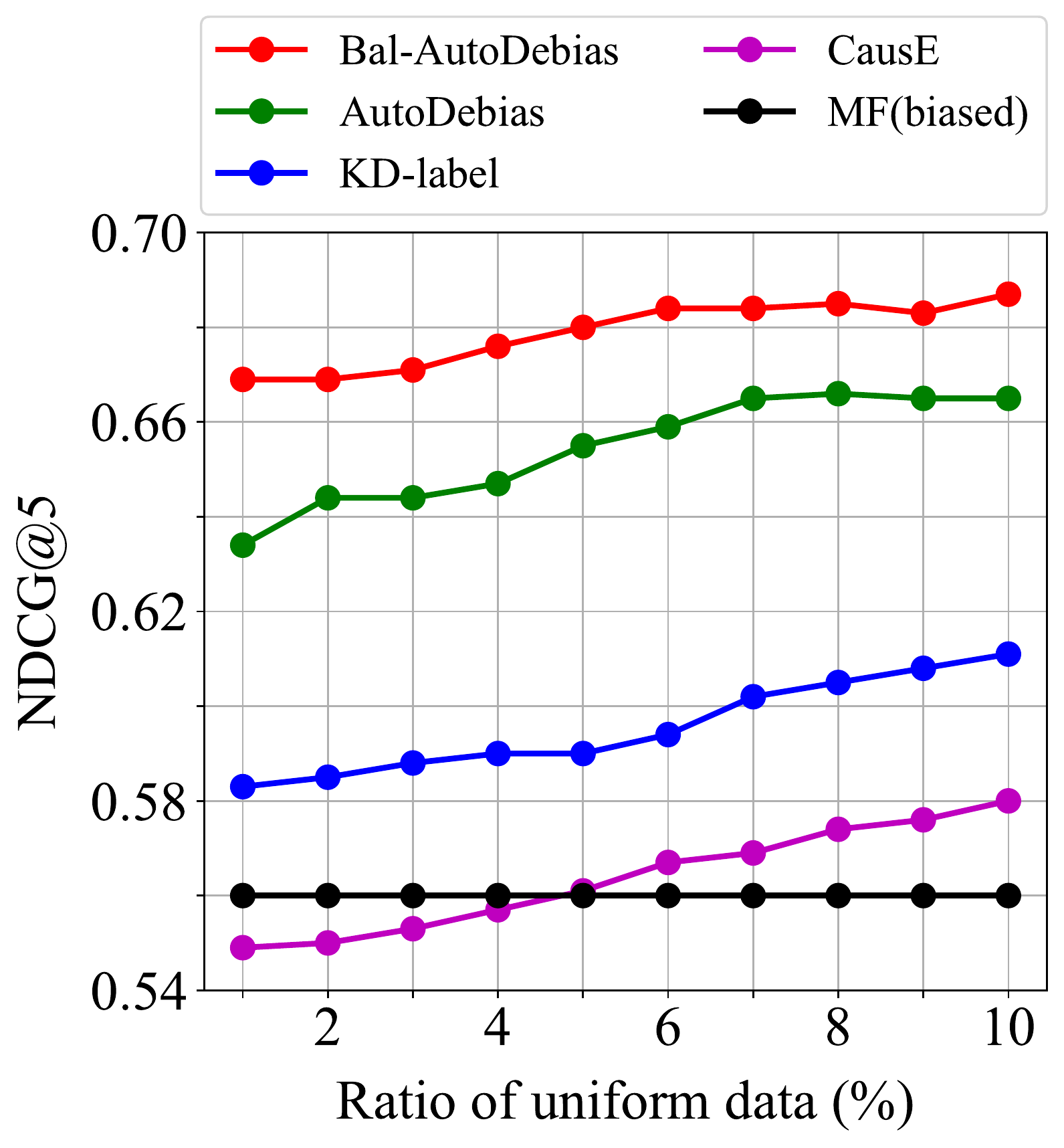}
\end{minipage}%
}%
\subfigure[\textsc{Music}-NDCG@10]{
\begin{minipage}[t]{0.495\linewidth}
\centering
\includegraphics[width=1\textwidth]{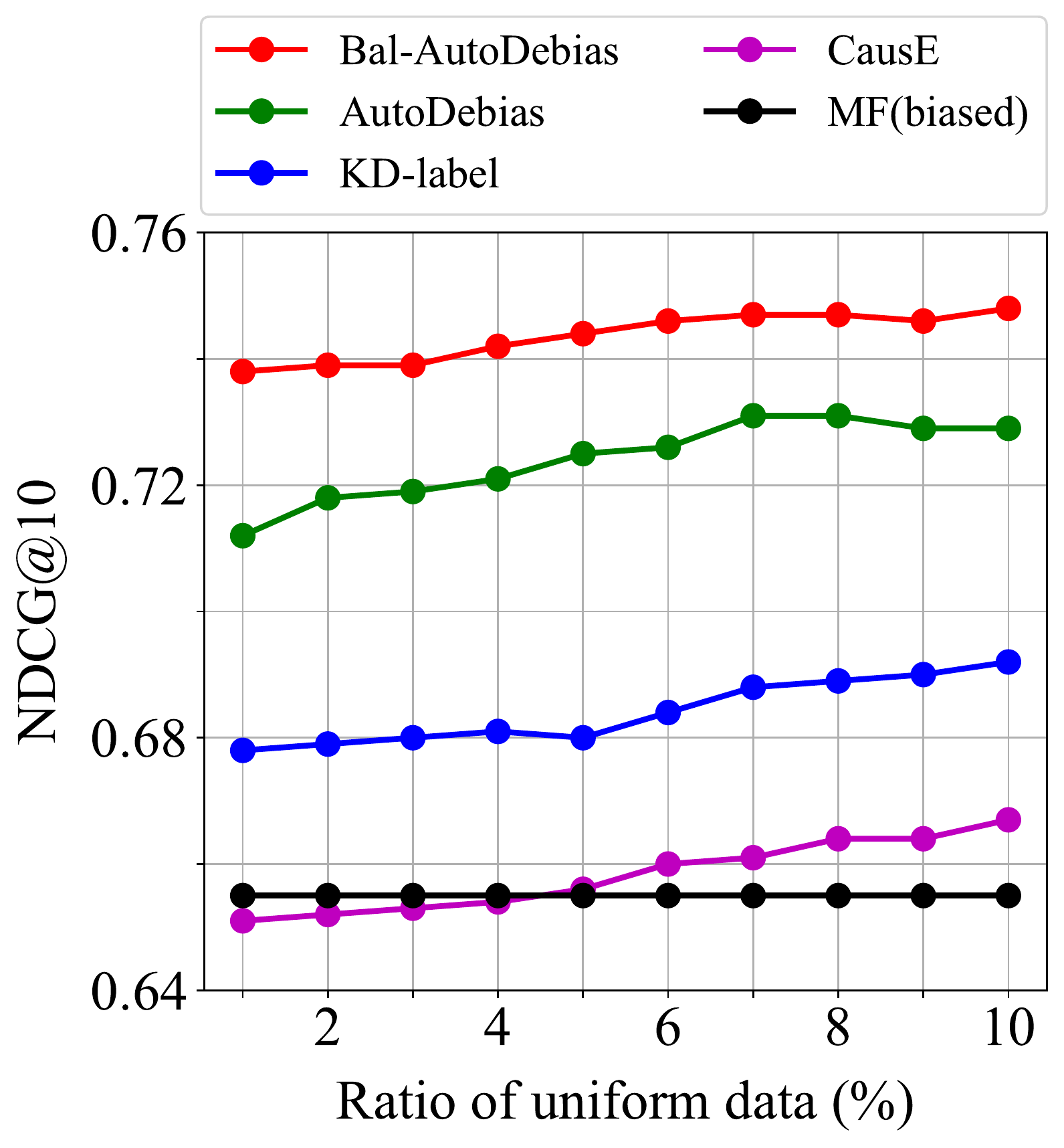}
\end{minipage}%
}%
\centering
\vspace{-8pt}

\centering
\subfigure[\textsc{Coat}-NDCG@5]{
\begin{minipage}[t]{0.495\linewidth}
\centering
\includegraphics[width=1\textwidth]{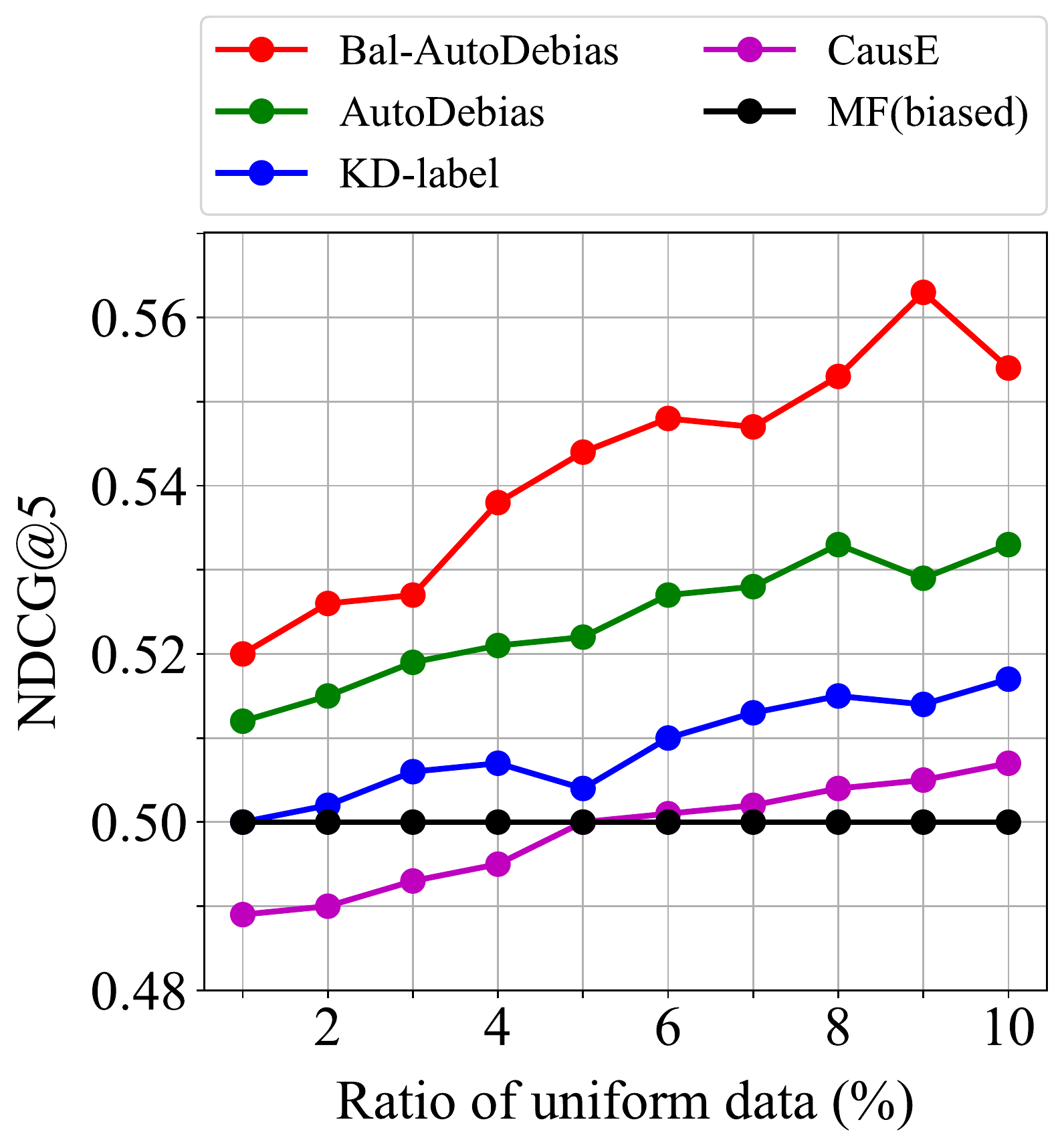}
\end{minipage}%
}%
\subfigure[\textsc{Coat}-NDCG@10]{
\begin{minipage}[t]{0.495\linewidth}
\centering
\includegraphics[width=1\textwidth]{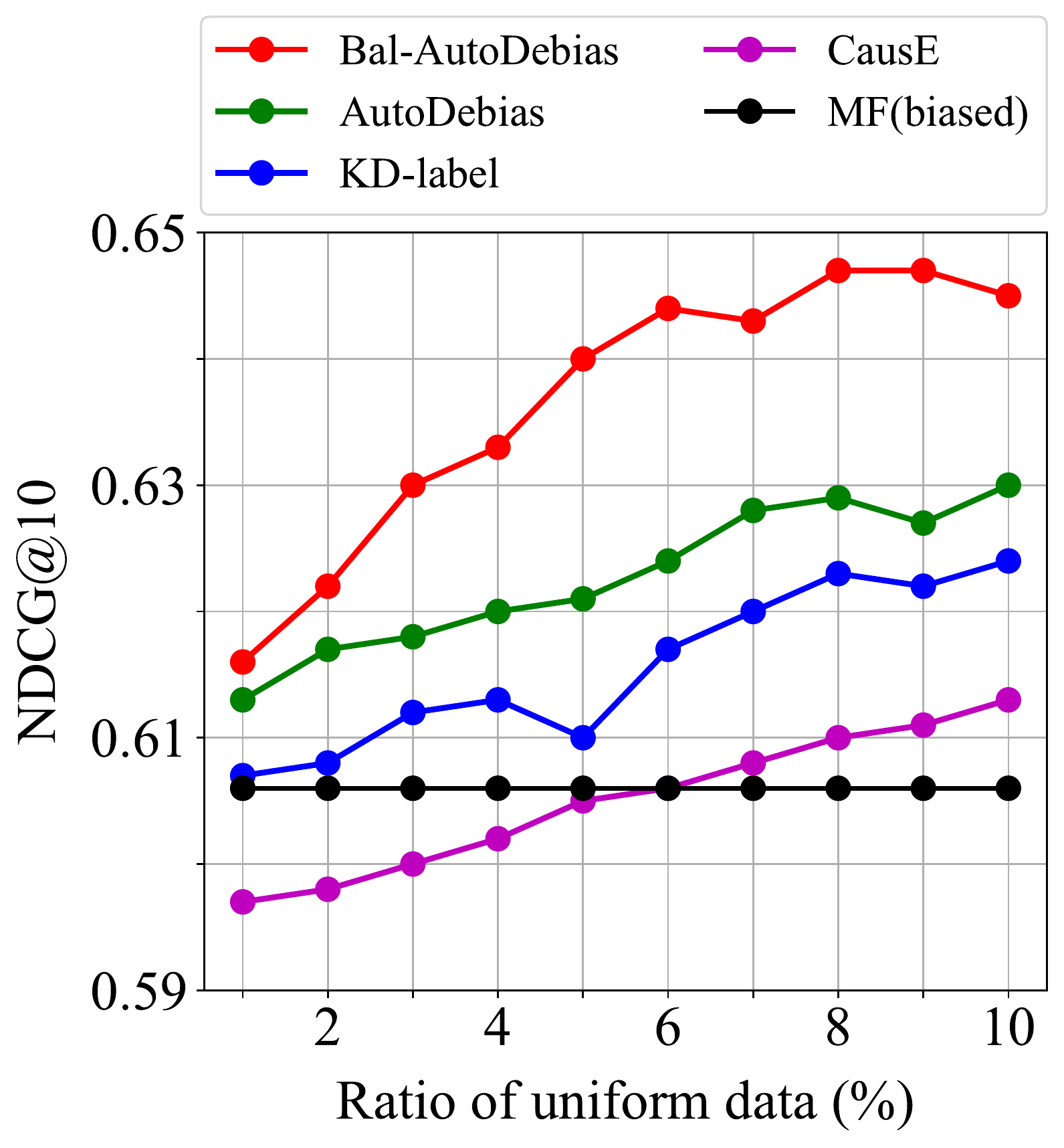}
\end{minipage}%
}%
\centering
\vspace{-8pt}
\caption{Effect of varying size of uniform data.}
\vspace{-6pt}
\label{fig4}
\end{figure}

\subsection{In-depth Analysis (RQ2)}
We further conduct an in-depth analysis by using the pre-trained prediction model parameters given by IPS, DR and AutoDebias as initialization in Alg. \ref{alg2}, respectively, to verify that the proposed Bal-methods can be effectively applied to any existing debiasing methods. The results are presented in Table \ref{tab3}. We find that all Bal-methods show significant performance improvement in all metrics compared to the pre-trained prediction models. Notably, applying the Bal-methods to any initialized predictions can stably boost the performance compared with AutoDebias on \textsc{Coat}, which can be explained by the possible presence of unobserved confounding and model misspecification in the biased data, while our method can mitigate the potential bias via a model-agnostic manner.

\subsection{Ablation Study (RQ3)}
To explore the impact of the proposed balancing model on the debiasing performance, we conduct ablation experiments using varying regularization hyperparameters $\lambda$ for trade-offs between the original loss estimation and the correction due to the unobserved confounding. Note that when $\lambda=0$, the globally optimal balancing weights equal to $1/|\cD|$ with maximum entropy, degenerating to the standard IPS, DR and AutoDebias. We tune $\lambda$ in \{0, $2^{-9}$, $2^{-6}$, $2^{-3}$, 1\} on Bal-IPS, Bal-DR and Bal-AutoDebias, and the results are shown in Figure \ref{fig:confounder}, where the black dashed line is used as the most competitive baseline for reference. We find that the AUC and NDCG@K of all methods first increase and then decrease with the increasing constraint strength, with optimal performance around $\lambda=2^{-6}$. This is interpreted as the best tradeoff between estimated loss and unobserved confounding. All methods using $\lambda>0$ stably outperform the standard AutoDebias and the case without considering unobserved confounding, i.e., $\lambda=0$, so it can be concluded that the proposed balancing model plays an important role in the debiasing.

\subsection{Exploratory Analysis (RQ4)}
{\bf Effect of balancing model selections.} We further explore the effect of model selections on the balanced weights to the debiasing performance. Specifically, we take different combinations of MF and NCF as balancing models for $w_{u, i, 1}$ on $\cD$ and $w_{u, i, 2}$ on $\mathcal{B}$, and the results are shown in Table \ref{tab4}. The performance can be significantly improved when NCF and MF are used to model $w_{u, i, 1}$ and $w_{u, i, 2}$, respectively. We argue that the main reason is that $|\cD| \gg |\mathcal{B}|$, leading to a reasonable reparameterization of $w_{u, i, 1}$ using deep models (e.g., NCF), and $w_{u, i, 2}$ using simple models (e.g., MF).
\smallskip \\
{\bf Effect of uniform data size.} Figure \ref{fig4} shows the sensitivity of the debiasing methods to the size of the uniform data ranging from 1\% to 10\%. We find that the proposed Bal-AutoDebias stably outperforms the existing methods for varying sizes of unbiased ratings. For the previous methods, AutoDebias has a more competitive performance compared with KD-label and CausE. When providing with a small size (e.g., 1\%) of the unbiased ratings, CausE performs even worse than the biased MF, while Bal-AutoDebias achieves the optimal performance. Compared with AutoDebias, our methods make significant improvements on both NDCG@5 and NDCG@10, validating the effectiveness of the proposed balancing learning.


\vspace{-0.2cm} 
\section{Related Work}
\noindent
{\bf Debiased Recommendation.} Recommender algorithms are often trained based on the historical interactions. However, the historical data cannot fully represent the user's true preference~\cite{Chen-etal2020, Wu-etal2022}, because user behavior is affected by various factors, such as conformity~\cite{DBLP:conf/recsys/LiuCY16} and item popularity~\cite{zhang2021causal}, etc. Many methods were developed for achieving unbiased learning, aiming to capture the true user preferences with biased data. For example, \citep{Schnabel-Swaminathan2016} noticed the missing data problem in RS and recommended using the IPS strategy to remove the bias,  
 \cite{Wang-Zhang-Sun-Qi2019} designed a doubly robust (DR) loss and 
 suggested adopting the joint learning method for model training.  Subsequently, several approaches enhanced the DR method by pursuing a better bias-variance trade-off~\cite{MRDR, Dai-etal2022}, leveraging parameter sharing and multi-task learning technique~\cite{Multi_IPW, ESMM, wang2022escm}, combing a small uniform dataset~\cite{bonner2018causal,Chen-etal2021,liu2020general,Wang-etal2021}, addressing the problem of small propensities and weakening the reliance on extrapolation~\cite{SDR}, and reducing bias and variance simultaneously when the imputed errors are less accurate~\cite{TDR}.  In addition, \cite{li-etal-MR2023} proposed a multiple robust learning method that allows the use of multiple candidate propensity and imputation models and is unbiased when any of the propensity or imputation models is accurate.   
 \cite{Chen-etal2020, Wu-etal2022} reviewed the recent progress in debiased recommendation. To mitigate the effects of unobserved confounding, \cite{ding2022addressing} proposed an adversarial learning method that uses only biased ratings. Unlike the existing methods, this paper combats the effect of unmeasured confounding with a small uniform dataset to achieve exact unbiasedness.   

  

\smallskip \noindent 
{\bf {Causal Inference under Unmeasured Confounding.}}   
Unmeasured confounding is a difficult problem in causal inference and the main strategies for addressing it can be divided into two classes~\cite{bareinboim2016causal,hunermund2019causal,kallus2018removing,li2021improving}. One is the sensitivity analysis~\cite{christopher2002identification, Kallus-Zhou2018, Rosenbaum2020} 
 that seeks bounds for the true causal effects with datasets suffering from unmeasured confounders. The other class methods aim to obtain unbiased causal effect estimators by leveraging some auxiliary information, such as  
      instrument variable methods~\cite{IV-1996, Hernan-Robins2020}, front door adjustment~\cite{Pearl-2009}, and negative control~\cite{NegControl}. In general, finding a reliable instrument variable or a mediator that satisfies the front door criterion~\cite{Hernan-Robins2020, imbens2020potential} is a challenging task in practice.             
Different from these methods based on an observational dataset, this paper considers a more practical scenario in debiased recommendations, i.e., addressing unmeasured confounding by fully exploiting the unbiasedness property of a small uniform dataset.  
\vspace{-0.2cm} 
\section{Conclusion}
This paper develops a method for balancing unobserved confounding with few unbiased ratings. We first show theoretically that previous methods that simply using unbiased ratings to select propensity and imputation model parameters is not sufficient to combat the effects of unobserved confounding and model misspecification. We then propose a balancing optimization training objective, and further propose a model-agnostic training algorithm to achieve the training objective using reparameterization techniques. The balancing model is alternately updated with the prediction model to combat the effect of unobserved confounding. We conduct extensive experiments on two real-world datasets to demonstrate the superiority of the proposed approach. To the best of our knowledge, this is the first paper using a few unbiased ratings to combat the effects of unobserved confounding in debiased recommendations. For future works, we will derive theoretical generalization error bounds for the balancing approaches, as well as explore more effective ways to leverage the unbiased ratings to enhance the debiasing performance of the prediction models.

\section{Acknowledgments}
This work was supported by the National Key R\&D Program of China (No. 2018YFB1701500 and No. 2018YFB1701503). 

\bibliographystyle{ACM-Reference-Format}
\bibliography{reference}

\end{document}